\documentclass{sig-alternate-05-2015}

\newfont{\mycrnotice}{ptmr8t at 7pt}
\newfont{\myconfname}{ptmri8t at 7pt}
\let\crnotice\mycrnotice%
\let\confname\myconfname%

\permission{Permission to make digital or hard copies of all or part of this work for personal or classroom use is granted without fee provided that copies are not made or distributed for profit or commercial advantage and that copies bear this notice and the full citation on the first page. Copyrights for components of this work owned by others than the author(s) must be honored. Abstracting with credit is permitted. To copy otherwise, or republish, to post on servers or to redistribute to lists, requires prior specific permission and/or a fee. Request permissions from permissions@acm.org.}
\CopyrightYear{2017} 
\setcopyright{acmcopyright}
\conferenceinfo{CODASPY'17,}{March 22-24, 2017, Scottsdale, AZ, USA}
\isbn{978-1-4503-4523-1/17/03}\acmPrice{\$15.00}
\doi{http://dx.doi.org/10.1145/3029806.3029808}

\usepackage{amsmath}
\usepackage{amssymb}
\usepackage{latexsym}
\usepackage{array}
\usepackage{mathbbol}
\usepackage{stmaryrd}
\usepackage{paralist}
\usepackage{graphicx} 
\usepackage{subcaption}
\usepackage{xspace}
\usepackage{url}
\usepackage{tabulary,booktabs,multirow}
\usepackage{algorithm}
\usepackage{tikz}
\usetikzlibrary{arrows}
\usepackage{algpseudocode}
\usepackage{pifont}
\usepackage{listings}

\newcommand{\four}{\mathsf{4}}
\newcommand{\fourk}{\mathsf{4}_k}

\newcommand{\set}[1]{\ensuremath{\left\{#1\right\}}} 
\newcommand{\semt}[2]{\nu_{#2}(#1)}
\newcommand{\semp}[2]{\delta_{#2}(#1)}

\newcommand{\pand}{\mathbin{\wedge_{\rm p}}}
\newcommand{\por}{\mathbin{\vee_{\rm p}}}
\newcommand{\pnot}{\mathbin{\neg}}
\newcommand{\dbd}{\mathbin{\sim}}
\newcommand{\dnot}{\mathop{\neg}}

\newcommand{\dov}{\mathbin{\mathsf{do}}}

\newcommand{\pov}{\mathbin{\mathsf{po}}}

\newcommand{\ooa}{\mathbin{\mathsf{ooa}}}

\newcommand{\ptacl}{PTaCL\xspace}
\newcommand{\ptacle}{PTaCL(E)\xspace}

\newcommand{\maxf}{\phi_{\max}}
\newcommand{\minf}{\phi_{\min}}
\newcommand{\minop}{\curlywedge}
\newcommand{\maxop}{\curlyvee}

\newcommand{\selop}[2]{S_{#1}^{#2}}

\newcommand{\bigmaxop}{\bigcurlyvee}
\newcommand{\bigminop}{\bigcurlywedge}

\newcommand{\jand}{\wedge_{\rm e}}

\newcommand{\tord}{\leqslant_t}
\newcommand{\kord}{\leqslant_k}
\newcommand{\btor}{\mathbin{\vee_{\rm b}}}
\newcommand{\btand}{\mathbin{\wedge_{\rm b}}}
\newcommand{\bkor}{\mathbin{\oplus_{\rm b}}}
\newcommand{\bkand}{\mathbin{\otimes_{\rm b}}}
\newcommand{\bneg}{\mathop{\neg}}
\newcommand{\bnegt}{\mathop{\neg}}
\newcommand{\bnegk}{\mathop{-}}

\newcommand{\btran}[1]{\mathop{\sim_{#1}}}
\newcommand{\bcyc}{\mathop{\diamond}}

\newcommand{\xmin}{\underline{x}}
\newcommand{\xmax}{\overline{x}}

\newcommand{\transp}[2]{\left(#1\, #2\right)}

\newcommand{\mlogicneg}{\mathop{\dagger}}
\newcommand{\mlogiccycle}{\mathop{\diamond}}

\newcommand{\mor}{\vee_t}
\newcommand{\mand}{\mathbin{\wedge_t}}
\newcommand{\mrev}{\mathop{\updownarrow}}

\newcommand{\ptacll}[1]{\text{PTaCL}_{#1}^{\leqslant}\xspace}

\newtheorem{example}{Example}
\newtheorem{theorem}{Theorem}

\newtheorem{remark}{Remark}
\newtheorem{lemma}{Lemma}
\newtheorem{corollary}{Corollary}
\newtheorem{proposition}{Proposition}

\begin{document}

\title{Canonical Completeness in Lattice-Based Languages\\ for Attribute-Based Access Control}
\numberofauthors{2}
\author{
\alignauthor
Jason Crampton\\
    \affaddr{Royal Holloway University of London}\\
    \affaddr{Egham, TW20 0EX, United Kingdom} \\
    \email{jason.crampton@rhul.ac.uk}
\alignauthor
Conrad Williams \\
    \affaddr{Royal Holloway University of London}\\
    \affaddr{Egham, TW20 0EX, United Kingdom} \\
    \email{conrad.williams.2010@live.rhul.ac.uk}
}


\date{}
\maketitle

\begin{abstract}
The study of canonically complete attribute-based access control (ABAC) languages is relatively new. 
A canonically complete language is useful as it is functionally complete and provides a ``normal form'' for policies.
However, previous work on canonically complete ABAC languages requires that the set of authorization decisions is totally ordered, which does not accurately reflect the intuition behind the use of the allow, deny and not-applicable decisions in access control.
A number of recent ABAC languages use a fourth value and the set of authorization decisions is partially ordered.
In this paper, we show how canonical completeness in multi-valued logics can be extended to the case where the set of truth values forms a lattice.
This enables us to investigate the canonical completeness of logics having a partially ordered set of truth values, such as Belnap logic, and show that ABAC languages based on Belnap logic, such as PBel, are not canonically complete.
We then construct a canonically complete four-valued logic using connections between the generators of the symmetric group (defined over the set of decisions) and unary operators in a canonically suitable logic.
Finally, we propose a new authorization language $\ptacll{\four}$, an extension of \ptacl, which incorporates a lattice-ordered decision set and is canonically complete.
We then discuss how the advantages of $\ptacll{\four}$ can be leveraged within the framework of XACML.
\end{abstract}

%
%


 \section{Introduction}\label{sec:intro}

 \emph{Access control} is one of the most important security services in multi-user computer systems, providing a mechanism for constraining the interaction between (authenticated) users and protected resources.
 Generally, access control is implemented by an authorization service, which includes an \emph{authorization decision function} for deciding whether a user request to access a resource (an ``access request'') should be permitted or not.
 In its simplest form an authorization decision function either returns an $\sf{allow}$ or $\sf{deny}$ decision.
 
 Most implementations of access control use \emph{authorization policies}, where a user request to access a resource is evaluated with respect to a policy that defines which requests are authorized.
 Many recent languages for the specification of authorization policies are designed for ``open'', distributed systems (rather than the more traditional ``closed'', centralized systems in which the set of users was assumed to be known in advance).
 Such languages do not necessarily rely on user identities to specify policies; instead, policies are defined in terms of other user and resource attributes. 
 The most widely used attribute-based access control (ABAC) language is XACML~\cite{moses:XACML05,rissanen:XACML12}.
 However, XACML suffers from poorly defined and counterintuitive semantics~\cite{li:access09,ni:dalg09}, and is inconsistent in its articulation of policy evaluation.
 \ptacl is a more formal language for specifying authorization policies~\cite{crampton:ptacl12}, providing a concise syntax for policy targets and precise semantics for policy evaluation.
 
 Crampton and Williams~\cite{crampton:decision16} recently introduced the notion of \emph{canonical completeness} for ABAC languages, showing that XACML and \ptacl are not canonically complete and developing a variant of \ptacl that is canonically complete.
 These results apply to languages that support three decision values, which are assumed to be totally ordered.
 However, there are certain situations where it is useful to have four decisions available, and some languages, such as PBel~\cite{bruns:belnap11}, BelLog~\cite{tsankov:dece14} and Rumpole~\cite{marinovic:rumpole14}, use four decisions, which are partially ordered.
  
 In this paper, we extend existing results to languages that support four decision values, which need not be totally ordered.
 We show that PBel~\cite{bruns:belnap11}, perhaps the best-known four-valued ABAC language, is not canonically complete.
 We then develop a canonically complete ABAC language, based on \ptacl syntax and semantics.
 The language is abstract, but its operators could be implemented as combining algorithms in XACML, thereby leveraging the features that XACML provides for specifying attribute-based requests and targets, the evaluation of targets with respect to requests, and the storage and evaluation of policies.

In Section~\ref{sec:background} we discuss background material and related work, which provides us with the primary motivation for this paper: to develop a canonically complete $4$-valued logic to support a tree-structured authorization language.
The main contributions of this work are:
\begin{itemize}
 \item to extend Jobe's work on canonical completeness in multi-valued logics to the case where the set of truth values forms a lattice (Section~\ref{sec:canonical-completeness-in-lattice-based-logic});
 \item to establish that existing $4$-valued logics are not canonically complete (Section~\ref{sec:canon-complete-belnap});
 \item to construct a canonically complete $4$-valued logic (Section~\ref{sec:canon-suitable-set});
 \item to construct a $4$-valued, canonically complete authorization language for ABAC (Section~\ref{sec:canonically-complete-abac-lang-4val}).
\end{itemize}

We conclude the paper with a summary of our contributions and a discussion of future work.
\section{Background and Related Work}
\label{sec:background}
 In this section, we summarize background material and related work, including tree-structured ABAC languages, canonical completeness, and four-valued languages for ABAC, thereby providing motivation for the work in the remainder of the paper.

\subsection{Completeness in Multi-valued Logics}\label{sec:completeness-in-mvalued-logics}
 
%
%
%

Let $V$ be a set of truth values.
The set of formulae $\Phi(L)$ that can be written in a (multi-valued) propositional logic $L = (V,{\sf Ops})$ is defined by $V$ and the set of operators $\sf Ops$.  
For brevity, we will write $L$ when $V$ and $\sf Ops$ are obvious from context.

Let $V$ be a totally ordered set of $m$ truth values, $\set{0,\dots,m-1}$, with $0 < 1 < \dots  < m-1$.
Then we say $L = (V,{\sf Ops})$ is \emph{canonically suitable} if and only if there exist two formulas $\maxf$ and $\minf$ of arity $2$ in $\Phi(L)$ such that $\maxf(x,y)$ returns $\max\set{x,y}$ and $\minf(x,y)$ returns $\min\set{x,y}$.
We will usually write $\maxf$ and $\minf$ using the infix operators $\maxop$ and $\minop$ respectively.

\begin{example}
Standard propositional logic with truth values $0$ and $1$, and operators $\vee$ and $\neg$, representing disjunction and negation, respectively, is canonically suitable: 
$\maxf(x,y)$ is simply \mbox{$x \vee y$}, while $\minf(x,y)$ is $\neg(\neg x \vee \neg y)$ (that is, conjunction).
\end{example}

A function $f: V^n \rightarrow V$ is completely specified by a truth table containing $n$ columns and $m^n$ rows. 
However, not every truth table can be represented by a formula in a given logic $L = (V,{\sf Ops})$.
$L$ is said to be \emph{functionally complete} if for every function $f: V^n \rightarrow V$, there is a formula $\phi \in \Phi(L)$ of arity $n$ whose evaluation corresponds to the truth table.
In Section~\ref{sec:tree-structured-languages-for-abac}, we explain why we may regard a tree-structured authorization language as a logic defined by a set of decisions and the set of policy-combining operators.
In this sense, XACML is not functionally complete~\cite{crampton:decision16}, while \ptacl~\cite{crampton:ptacl12} and PBel are~\cite{bruns:belnap11}.

A \emph{selection operator} $\selop{(a_1,\dots,a_n)}{j}$ is an $n$-ary operator defined as follows:
\[
\selop{(a_1,\dots,a_n)}{j}(x_1,\dots,x_n) = 
\begin{cases}
j & \text{if  $(x_1,\dots,x_n) = (a_1,\dots,a_n)$}, \\
0 & \text{otherwise}.
\end{cases}
\]
We will write ${\bf a}$ to denote the tuple $(a_1,\dots,a_n) \in V^n$ when no confusion can occur.
Note that $\selop{{\bf a}}{0}$ is the same for all ${\bf a} \in V^n$, and $\selop{{\bf a}}{0}({\bf x}) = 0$ for all ${\bf x} \in V^n$.
Illustrative examples of binary selection operators (for a $4$-valued logic) are shown in Figure~\ref{fig:ex-selection-operators}.

\begin{figure}[h]
 \[
  \begin{array}{r|@{~}r@{~}r@{~}r@{~}r}
     \selop{(0,2)}{1} & 0 & 1 & 2 & 3 \\
   \hline
    0 & 0 & 0 & 1 & 0 \\
    1 & 0 & 0 & 0 & 0 \\
    2 & 0 & 0 & 0 & 0 \\
    3 & 0 & 0 & 0 & 0 \\
  \end{array}
  \qquad
  \begin{array}{r|@{~}r@{~}r@{~}r@{~}r}
     \selop{(1,1)}{2} & 0 & 1 & 2 & 3 \\
   \hline
    0 & 0 & 0 & 0 & 0 \\
    1 & 0 & 2 & 0 & 0 \\
    2 & 0 & 0 & 0 & 0 \\
    3 & 0 & 0 & 0 & 0 \\
  \end{array}
  \qquad
  \begin{array}{r|@{~}r@{~}r@{~}r@{~}r}
     \selop{(3,0)}{3} & 0 & 1 & 2 & 3 \\
   \hline
    0 & 0 & 0 & 0 & 0 \\
    1 & 0 & 0 & 0 & 0 \\
    2 & 0 & 0 & 0 & 0 \\
    3 & 3 & 0 & 0 & 0 \\
  \end{array}
 \]
 \caption{Selection operators $\selop{(0,2)}{1}$, $\selop{(1,1)}{2}$ and $\selop{(3,0)}{3}$ }\label{fig:ex-selection-operators}
\end{figure}

Selection operators play a central role in the development of canonically complete logics because an arbitrary function $f : V^n \rightarrow V$ can be expressed in terms of selection operators.
Consider, for example, the function 
\[
 f(x,y) = %
  \begin{cases}
   1 & \text{if $x = 0$, $y = 2$}, \\
   2 & \text{if $x = y = 1$}, \\
   3 & \text{if $x = 3$, $y = 0$}, \\
   0 & \text{otherwise}.
  \end{cases}
\]
Then it is easy to confirm that 
\[
 f(x,y) \equiv \selop{(0,2)}{1}(x,y) \maxop \selop{(1,1)}{2}(x,y) \maxop \selop{(3,0)}{3}(x,y).
\]
Moreover, $\selop{(a,b)}{c}(x,y) \equiv \selop{a}{c}(x) \minop \selop{b}{c}(y)$ for any $a,b,c,x,y \in V$.
Thus, 
\[
 f(x,y) \equiv (\selop{0}{1}(x) \minop \selop{2}{1}(y)) \maxop (\selop{1}{2}(x) \minop \selop{1}{2}(y)) \maxop (\selop{3}{3}(x) \minop \selop{0}{3}(y)) 
\]
%
In other words, we can express $f$ as the ``disjunction'' ($\maxop$) of ``conjunctions'' ($\minop$) of unary selection operators.

More generally, given the truth table of function $f : V^n \rightarrow V$, we can write down an equivalent function in terms of selection operators.
Specifically, let 
\[ 
  A = \set{{\bf a} \in V^n : f({\bf a}) > 0};
\] 
then, for all ${\bf x} \in V^n$,
\[
  f({\bf x}) = \bigmaxop_{{\bf a} \in A} \selop{\bf a}{f({\bf x})}(\bf x).
\]
%
Jobe established a number of results connecting the functional completeness of a logic with the unary selection operators, summarized in the following theorem.


\begin{theorem}[Jobe~\protect{\cite[Theorems 1, 2; Lemma 1]{jobe:functional62}}]\label{thm:jobe-functional-completeness}
 A logic $L$ is functionally complete if and only if each unary selection operator is equivalent to some formula in $L$.
\end{theorem}

%

The \emph{normal form} of formula $\phi$ in a canonically suitable logic is a formula $\phi'$ that has the same truth table as $\phi$ and has the following properties:
 \begin{itemize}
  \item the only binary operators it contains are $\maxop$ and $\minop$;
  \item no binary operator is included in the scope of a unary operator;
  \item no instance of $\maxop$ occurs in the scope of the $\minop$ operator.
 \end{itemize}
In other words, given a canonically suitable logic $L$ containing unary operators $\sharp_1,\dots,\sharp_\ell$, a formula in normal form has the form
\[
 \bigmaxop_{i=1}^r \bigminop_{j=1}^{s} \sharp_{i,j} x_{i,j}
\]
where $\sharp_{i,j}$ is a unary operator defined by composing the unary operators in $\sharp_1,\dots,\sharp_\ell$.
In the usual $2$-valued propositional logic with a single unary operator (negation) this corresponds to disjunctive normal form.

A canonically suitable logic is \emph{canonically complete} if every unary selection operator can be expressed in normal form.
It is known that there are canonically suitable $3$-valued logics that are:%
\begin{inparaenum}[(i)]
 \item not functionally complete~\mbox{\cite{jobe:functional62,luka:phil30}};
 \item functionally complete but not canonically complete~\cite[Theorem 4]{jobe:functional62}; and 
 \item canonically complete (and hence functionally complete)~\cite[Theorem 6]{jobe:functional62}.
\end{inparaenum}

Jobe defined a canonically complete $3$-valued logic~\cite{jobe:functional62}.
The operators and the construction of the unary selection operators using these operators are given in Appendix~\ref{app:jobe-3-val}.
The expression for $f$ above could be expressed in normal form, providing we could find suitable unary operators for a $4$-valued logic.
In Section~\ref{sec:unary-ops-totally-ordered-logics} we explain how to produce a suitable set of unary operators for an $m$-valued logic.
    
  \subsection{Tree-structured Languages for ABAC}\label{sec:tree-structured-languages-for-abac}
  
  Let $D$ be a set of \emph{authorization decisions}.
  Typically, we assume $D$ contains the values $0$, $1$ and $\bot$ representing ``deny'', ``allow'' and ``not-applicable'', respectively.
  We call $0$ and $1$ \emph{conclusive} decisions.
  Let $\oplus$ be an associative binary operator defined on $D$ and $-$ be a unary operator defined on $D$.
  Then
   \begin{itemize}
    \item an \emph{atomic policy} is a pair $(t,d)$, where $d$ is a decision in $D$ and $t$ is a \emph{target predicate};
    \item an atomic policy is a \emph{policy};
    \item if $p$ and $p'$ are policies, then $(t,p \oplus p')$ and $(t,-p)$ are policies.
   \end{itemize}
  We will write $p$ to denote the policy $(\mathsf{true},p)$.
  
  The first stage in policy evaluation for a request $q$ is to determine whether a policy is ``applicable'' to $q$ or not.
  Every (well-formed) request $q$, allows us to assign a truth value to the target $t$.
  Specifically, $t$ may evaluate to {\sf true}, in which case the associated policy is \emph{applicable}; otherwise the policy is \emph{not applicable}.%
   \footnote{The evaluation of requests is not relevant to the exposition of this paper. XACML provides a means of specifying requests and targets and an evaluation architecture for determining whether a target is applicable or not.}
  Then, writing $\semt{t}{q}$ to denote the truth value assigned to $t$ by $q$ and $\semp{p}{q}$ to denote the decision assigned to policy $p$ for request $q$, we define:
  \begin{align*}
   \semp{t,p}{q} &= %
    \begin{cases}
     \semp{p}{q} & \text{if $\semt{t}{q} = 1$}, \\
     \bot & \text{otherwise};
    \end{cases} \\
   \semp{d}{q} &= d; \\
   \semp{-p}{q} &= -\semp{p}{q}; \\
   \semp{p \oplus p'}{q} &= \semp{p}{q} \oplus \semp{p'}{q}.
  \end{align*}
  
  It is easy to see that we may represent a policy as a tree.
  Hence, we describe policy languages of this nature as \emph{tree-structured}.
  The first stage in policy evaluation corresponds to labeling the nodes of tree applicable or not applicable.
  We then compute a decision for non-leaf nodes in the tree by combining the decisions assigned to their respective children.
  Figure~\ref{fig:policy-tree-example} shows the tree for the policy
  \[
  \big{(}t_6,(t_4,-(t_3,((t_1,1) \oplus_1 (t_2,0)))) \oplus_2 (t_5,0)\big{)}
  \]
  and the evaluation of that policy for a request $q$ such that $\nu_q(t_i) = {\sf true}$ for all $i$ except $i = 2$.
  
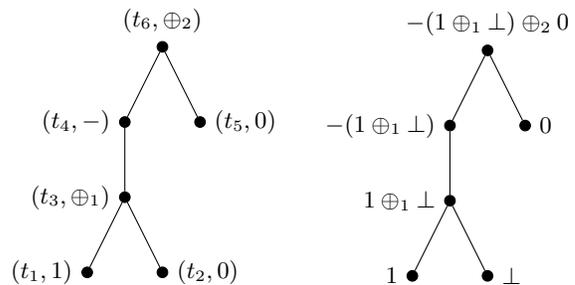
\begin{figure}[!ht]\centering
\begin{tikzpicture}[v/.style={circle,draw,fill=black,inner sep=0pt,minimum width=4pt},w/.style={fill=white}]
   	 \node[v,label=right:\textbf{$(t_2,0)$}] at (0,0) (a) {};
   	 \node[v,label=left:\textbf{$(t_1,1)$}] at (-1,0) (b) {};
   	 \node[v,label=left:\textbf{$(t_3,\oplus_1)$}] at (-0.5,1) (c) {};
   	 \node[v,label=left:\textbf{$(t_4,-)$}] at (-0.5,2) (d) {};
   	 \node[v,label=right:\textbf{$(t_5,0)$}] at (0.5,2) (e) {};
   	 \node[v,label=above:\textbf{$(t_6,\oplus_2)$}] at (0,3) (f) {};
   	 \draw (a) -- (c);
   	 \draw (b) -- (c);
	 \draw (c) -- (d);
	 \draw (d) -- (f);
	 \draw (e) -- (f);
    \end{tikzpicture}
    \quad
\begin{tikzpicture}[v/.style={circle,draw,fill=black,inner sep=0pt,minimum width=4pt},w/.style={fill=white}]
   	 \node[v,label=right:\textbf{$\bot$}] at (0,0) (a) {};
   	 \node[v,label=left:\textbf{$1$}] at (-1,0) (b) {};
   	 \node[v,label=left:\textbf{$1 \oplus_1 \bot$}] at (-0.5,1) (c) {};
   	 \node[v,label=left:\textbf{$-(1 \oplus_1 \bot)$}] at (-0.5,2) (d) {};
   	 \node[v,label=right:\textbf{$0$}] at (0.5,2) (e) {};
   	 \node[v,label=above:\textbf{$-(1 \oplus_1 \bot) \oplus_2 0$}] at (0,3) (f) {};
   	 \draw (a) -- (c);
   	 \draw (b) -- (c);
	 \draw (c) -- (d);
	 \draw (d) -- (f);
	 \draw (e) -- (f);
    \end{tikzpicture}
\caption{A policy tree and its evaluation}\label{fig:policy-tree-example}
\end{figure}

  There are several tree-structured ABAC languages in the literature, including the OASIS standard XACML, PBel and PTaCL~\cite{rissanen:XACML12,crampton:ptacl12,bruns:belnap11}.%
   \footnote{A number of policy algebras have also been defined, which have some similarities with tree-structured languages.  The semantics of a policy are defined in terms of sets of authorized and denied requests~\cite{BoViSa02,WiJa03,ni:dalg09,roa:fia09}, and policy operators are defined in terms of set operations such as intersection and union.}
  These languages differ to some extent in the choices of $D$ and the set of operators that are used.
  XACML, for example, defines several rule- and policy-combining algorithms (which may be regarded as binary operators), but no unary operators.%
   \footnote{An XACML rule is equivalent to an atomic policy.}
  PBel and PTaCL prefer to define a rather small set of operators: PTaCL defines a single binary operator and two unary operators, whereas PBel defines two binary operators and a single unary operator.
  XACML and PTaCL use a three-valued decision set comprising $0$, $1$ and $\bot$, to which PBel adds $\top$, which represents ``conflict''.

  The main difference between existing languages, however, is the extent to which they are complete in the senses defined in Section~\ref{sec:completeness-in-mvalued-logics}~\cite{crampton:decision16}.
  We summarize these differences in Table~\ref{tbl:properties-abac-languages}, where CS, FC and CC denote canonically suitable, functionally complete and canonically complete, respectively.
  In Section~\ref{sec:canon-complete-belnap}, we prove that PBel is canonically suitable but not canonically complete.

  \begin{table}[ht]\setlength{\extrarowheight}{1pt}\centering
   \begin{tabular}{l|rrr}
    \bf Language & \bf Decisions & \bf Unary Ops & \bf Binary Ops \\
    \hline
    XACML & $\set{0,1,\bot}$ & $0$ & $12$ \\
    PTaCL & $\set{0,1,\bot}$ & $2$ & $1$ \\
    PTaCL(E) & $\set{0,1,\bot}$ & $2$ & $1$ \\
    PBel & $\set{0,1,\bot,\top}$ & $1$ & $2$ \\
   \end{tabular}

   \vspace*{\baselineskip}
   
   \begin{tabular}{l|rrr}
    \bf Language & \bf CS? & \bf FC? & \bf CC? \\
    \hline
    XACML & No & No & No \\
    PTaCL & Yes & Yes & No \\
    PTaCL(E) & Yes & Yes & Yes \\
    PBel & ? & Yes & ? \\
   \end{tabular}
   \caption{Properties of ABAC languages}\label{tbl:properties-abac-languages}
  \end{table}
  
  \subsection{The Value of Canonical Completeness}

  One of the main difficulties with using a tree-structured language is writing the desired policy using the operators provided by the language.
  In particular, if it is not possible to express a policy using a single target and decision, the policy author must engineer the desired policy by combining sub-policies using the set of operators specified in the given language.
  This is a non-trivial task, in general. 
  Moreover, in XACML it may be impossible to write the desired policy due to its functional incompleteness.
  Thus, a policy author may be forced to write a policy that approximates the desired policy, which may lead to unintended or undesirable decisions for certain requests.
  
  An alternative approach, supported by XACML, is to define custom combining algorithms. 
  However, there is no guarantee that the addition of a new combining algorithm will make XACML functionally complete.
  Thus, more and more custom algorithms may be required over time.
  This, in turn, will make the design decisions faced by policy authors ever more complicated, thereby increasing the chances of errors and misconfigurations.
  
  In other words, we believe it is preferable to define a small number of operators having unambiguous semantics and providing functional completeness.
  A functionally complete ABAC language, such as \ptacl, can be used to construct any conceivable policy using the operators provided by the language.
  However, policy authors still face the challenge of finding the correct way to combine sub-policies using those operators to construct the desired policy.

 For example, \ptacl defines three policy operators $\pand, \neg$ and $\sim$.
 To express XACML's deny- and permit-overrides in \ptacl requires significant effort.
 For convenience, we introduce the operator $\por$:
\[
 d \por d' \stackrel{\rm def}{=} \neg((\neg d) \pand (\neg d')).
\]
It is then possible to show that
\begin{align*}
 d \pov d' &\equiv (d \por (\dbd d')) \pand ((\dbd d) \por d'),\ \text{and} \\
 d \dov d' &\equiv \neg((\neg d) \pov (\neg d')).
\end{align*}
The operators $\pov$ and $\dov$ are equivalent to the permit- and allow-overrides policy-combining algorithms in XACML.
As can be seen, the definitions of these operators in terms of the \ptacl operators are complex, and, more generally, it is a non-trivial task to derive such formulae. 
 
  Disjunctive normal form in propositional logic makes it trivial to write down a logical formula, using only conjunction, disjunction and negation, that is equivalent to an arbitrary Boolean function expressed in the form of a truth table.
  Similarly, a canonically complete ABAC language, such as \ptacle~\cite{crampton:decision16}, makes it possible to write down a policy in normal form from its decision table.
  In this paper, we show that there exist $4$-valued canonically complete logics in which the set of truth values forms a lattice.
  We discuss why and how this can simplify policy generation in Section~\ref{sec:ptacl4-operators-policies}.

  In addition, policies in normal form may be more efficient to evaluate. 
  Given a formula in a 3-valued logic expressed in normal form, any literal that evaluates to 0 causes the entire clause to evaluate to 0, while any clause evaluating to 1 means the entire formula evaluates to 1.
  In short, the time required for policy evaluation may be reduced in many cases. 
  (This is similar to the way in which algorithms such as first-applicable in XACML work: once an applicable policy is found policy evaluation terminates, even if there are additional policies that could be evaluated.)
  
  \subsection{The Value of a Fourth Decision}\label{sec:the-value-of-the-fourth-value}
  
The XACML 2.0 standard includes a fourth authorization decision ``indeterminate''~\cite{moses:XACML05}.
This is used to indicate errors have occurred during policy evaluation, meaning that a decision could not be reached.
The XACML 3.0 standard extends the definition of the indeterminate decision to indicate decisions that might have been reached, had evaluation been possible~\cite{rissanen:XACML12}.
However, the indeterminate decision is used in XACML 3.0 for more than reporting errors.
It is also used as a decision in the ``only-one-applicable'' combining algorithm, which returns indeterminate if two or more sub-policies are applicable.

More generally, a conflict decision is used in PBel (and languages such as Rumpole and BelLog) to indicate that two sub-policies return different conclusive decisions.
PBel is functionally complete.
In the remainder of this paper, we show that PBel is not canonically complete and then develop a canonically complete $4$-valued ABAC language.

\section{Lattice-based Multi-valued\\ Logics}\label{sec:canonical-completeness-lattices}

\renewcommand{\xmin}{\underline{v}}
\renewcommand{\xmax}{\overline{v}}

We first recall the definition of a lattice.
Suppose $(X,\leqslant)$ is a partially ordered set.
Then for a subset $Y$ of $X$, we say $u$ is an \emph{upper bound} of $Y$ if $y \leqslant u$ for all $y \in Y$.
We say $u'$ is a \emph{least upper bound} or \emph{supremum} of $Y$ if $u' \leqslant u$ for all upper bounds $u$ of $Y$.
Note that a least upper bound of $Y$ (if it exists) is unique.
We define \emph{greatest lower bound} or \emph{infimum} in an analogous way.
A lattice $(X,\leqslant)$ is a partially ordered set such that for all $x,y \in X$ there exists a least upper bound of $x$ and $y$, denoted $\sup\set{x,y}$, and a greatest lower bound of $x$ and $y$, denoted by $\inf\set{x,y}$.
The least upper bound of $x$ and $y$ is written as $x \vee y$ (the ``join'' of $x$ and $y$) and the greatest lower bound is written as $x \wedge y$ (the ``meet'' of $x$ and $y$).
If $(X,\leqslant)$ is a finite lattice, as we will assume henceforth, then $(X,\leqslant)$ has a maximum element (that is, a unique maximal element) and a minimum element.

In the remainder of this section we%
 \begin{inparaenum}[(i)]
   \item describe Belnap logic~\cite{belnap:logic77}, a well-known $4$-valued lattice-based logic;
   \item extend the definitions of canonical suitability, selection operators and canonical completeness to lattices; and
   \item show that Belnap logic and PBel are not canonically complete.
  \end{inparaenum}

\subsection{Belnap Logic}

\renewcommand{\four}{\mathsf{4}}

Belnap logic was developed with the intention of defining ways to handle inconsistent and incomplete information in a formal  manner.
It uses the truth values $0$, $1$, $\bot$, and $\top$, representing ``false'', ``true'', ``lack of information'' and ``too much information'', respectively.  
In the remainder of this paper, we will denote the four valued decision set $\set{\bot,0,1,\top}$ by $\four$.


The truth values $0$, $1$, $\bot$ and $\top$ have an intuitive interpretation in the context of access control: $0$ and $1$ are interpreted as the standard ``deny'' and ``allow'' decisions, $\bot$ is interpreted as ``not-applicable'' and $\top$ represents a conflict of decisions.
PBel is a $4$-valued tree-structured ABAC language~\cite{bruns:belnap11} based on Belnap logic.

The set of truth values in Belnap logic admits two orderings: a truth ordering $\tord$ and a knowledge ordering $\kord$.
In the truth ordering, $0$ is the minimum element and $1$ is the maximum element, while $\bot$ and $\top$ are incomparable indeterminate values.
In the knowledge ordering, $\bot$ is the minimum element, $\top$ is the maximum element while $0$ and $1$ are incomparable.
Both $(\four,\tord)$ and $(\four,\kord)$ are lattices, forming the interlaced bilattice illustrated in Figure~\ref{fig:four-structure}.

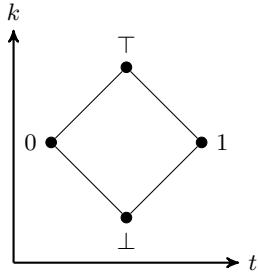
\begin{figure}[h]
\centering
 \begin{tikzpicture}[v/.style={circle,draw,fill=black,inner sep=0pt,minimum width=4pt},w/.style={fill=white}, axis/.style={ thick, ->, >=stealth'}]
     \draw[axis] (-1.5,-0.6)  -- (1.5,-0.6) node(xline)[right]
        {$t$};
    \draw[axis] (-1.5,-0.6) -- (-1.5,2.5) node(yline)[above] {$k$};
	 \node[v,label=below:\textbf{$\bot$}] at (0,0) (a) {};
	 \node[v,label=left:\textbf{$0$}] at (-1,1) (b) {};
	 \node[v,label=right:\textbf{$1$}] at (1,1) (c) {};
	 \node[v,label=above:\textbf{$\top$}] at (0,2) (d) {};
   	 \draw (a) -- (c);
   	 \draw (a) -- (b);
   	 \draw (b) -- (d);
   	 \draw (c) -- (d);
    \end{tikzpicture}
 \caption{The $4$ truth values in Belnap logic}\label{fig:four-structure}
\end{figure}

We write the meet and join in $(\four,\tord)$ as $\btand$ and $\btor$, respectively; and the meet and join in $(\four,\kord)$ as $\bkand$ and $\bkor$, respectively.
(We use the subscript $\rm b$ to differentiate the Belnap operators from the \ptacl operators $\pand$ and $\por$.)

We may interpret values in $V$ as operators of arity $0$ (that is, constants).
Then it is known that $L(\four,\set{\neg, \btand, \btor, \bkand, \bkor, \supset_{\rm b}, \bot,0,1,\top})$ is functionally complete~\cite[Theorem 12]{arieli:four98} and that $\set{\neg,\bkor,\supset_{\rm b},\bot}$ is a minimal functionally complete set of operators~\cite[Proposition 17]{arieli:four98}.
The truth tables for the binary operators $\btand$, $\btor$, $\bkand$, $\bkor$ and $\supset_{\rm b}$ are shown in Figure~\ref{fig:belnap-operators} (in Appendix~\ref{app:belnap-ops}).
The unary operator $\neg$ has the effect of switching the values $0$ and $1$, leaving $\bot$ and $\top$ fixed; in other words, it acts like ``classical'' negation.

\subsection{Canonical Completeness}\label{sec:canonical-completeness-in-lattice-based-logic}

Jobe's definition of canonical suitability for multi-valued logics assumes a total ordering on the set of truth values.
Given that Belnap logic~\cite{arieli:four98}, on which PBel is based, is a $4$-valued logic in which the set of truth values forms a lattice, we seek to extend the definition of canonical suitability to lattice-based logics.

Let $L$ be a logic associated with a lattice $(V,\leqslant)$ of truth values.
Then $L$ is \emph{canonically suitable} if and only if there exist in $L$ two formulas $\maxf$ and $\minf$ of arity $2$ such that $\maxf(x,y)$ returns sup$\set{x,y}$ and $\minf(x,y)$ returns inf$\set{x,y}$.
If a logic is canonically suitable, we will write $\maxf(x,y)$ and $\minf(x,y)$  using infix binary operators as $x \maxop y$ and $x  \minop y$, respectively.

\begin{remark}
The existence of $\sup\set{x,y}$ and $\inf \set{x,y}$ is guaranteed in a lattice; this is not true in general for partially ordered sets.
And for a totally ordered (finite) set,  $\sup\set{x,y} = \max\set{x,y}$ and $\inf\set{x,y} = \min\set{x,y}$, so our definitions are compatible with those of Jobe's for totally ordered sets of truth values.
\end{remark}

We now extend the definition of selection operators to a lattice-based logic.
%
Let $L$ be a logic associated with a lattice $(V,\leqslant)$ of truth values, with minimum truth value $\xmin$. 
Then, for ${\bf a} \in V^n$, the $n$-ary \emph{selection operator} $\selop{{\bf a}}{j}$ is defined as follows:
\[
\selop{{\bf a}}{j}({\bf x}) = 
\begin{cases}
j & \text{if ${\bf x} = {\bf a}$}, \\
\xmin & \text{otherwise.}
\end{cases}
\]
Note $\selop{{\bf a}}{\xmin}(x) = \xmin$ for all ${\bf a},{\bf x} \in V^n$.

The definitions for normal form and canonically completeness for lattice-based logics are identical to total-ordered logics.
Nevertheless, we reiterate the definitions here in the interests of clarity. 
%
The \emph{normal form} of formula $\phi$ in a canonically suitable logic is a formula $\phi'$ that has the same truth table as $\phi$ and has the following properties:
 \begin{itemize}
  \item the only binary operators it contains are $\maxop$ and $\minop$;
  \item no binary operator is included in the scope of a unary operator;
  \item no instance of $\maxop$ occurs in the scope of the $\minop$ operator.
 \end{itemize}
%
A canonically suitable logic is \emph{canonically complete} if every unary selection operator can be expressed in normal form.

\subsection{Completeness of Belnap Logic and PBel}\label{sec:canon-complete-belnap}
Having extended the definitions of canonical suitability, selection operators and canonical completeness to lattices, we now investigate how these concepts can be applied to Belnap logic~\cite{belnap:logic77}.
The meet and join operators of the two lattices $(\four,\tord)$ and $(\four,\kord)$ defined in Belnap logic are different.
Canonical suitability for a $4$-valued logic, defined as it is in terms of the ordering on the set of truth values, will thus depend on the ordering we choose on $\four$.
Consequently, the $\minop$ and $\maxop$ operators, along with the selection operators, will differ depending on the lattice that we choose.

In Section~\ref{sec:canonically-complete-abac-lang-4val}, we will argue in more detail for the use of a lattice-based ordering on $\four$ to support a tree-structured ABAC language.
For now, we state that we will use knowledge-ordered lattice $(\four,\kord)$.%
 \footnote{It is worth noting that results analogous to those presented in this paper can be obtained for the truth ordering.
	   Results for a total ordering on $\four$ can be derived using existing methods~\cite{crampton:decision16,jobe:functional62}.}
The intuition is that the minimum value in this lattice is $\bot$ (rather than $0$ in $(\four,\tord)$) and that this value should be the default value for a policy (being returned when the policy is not applicable to a request).
In the interests of brevity, we will henceforth write $\fourk$, rather than $(\four,\kord)$.

It follows from the functional completeness of $\set{\pnot,\bkor,\supset_{\rm b},\bot}$ that $L(\fourk,\set{\pnot,\bkor,\supset_{\rm b},\bot})$ is canonically suitable.
A similar argument applies to $\set{\neg,\btand,\supset_{\rm b},\bot,\top}$, the set of operators used in PBel. 

As $\bot$ is the minimum truth value in the lattice $\fourk$, the $n$-ary selection operator $\selop{\bf a}{j}$ for $\fourk$ is defined by the following function:
\[
 \selop{{\bf a}}{j}({\bf x}) = 
  \begin{cases}
   j & \quad \text{if } {\bf x} = {\bf a}, \\
   \bot & \quad \text{otherwise.}
  \end{cases}
\]
Examples of selection operators are shown in Figure~\ref{fig:selection-operators}.

\begin{figure}[h]
 \[
  \begin{array}{c|c|c}
   x & \selop{0}{0} & \selop{1}{\top} \\
  \hline
  \bot & \bot & \bot \\ 
  0 & 0 & \bot \\
  1 & \bot &  \top\\
  \top & \bot &  \bot\\
  \end{array}
  \qquad
  \begin{array}{c|cccc}
     \selop{(1,\top)}{0} & \bot & 0 & 1 & \top \\
   \hline
    \bot & \bot & \bot & \bot & \bot \\
    0 & \bot & \bot & \bot & \bot \\
    1 & \bot & \bot & \bot & 0 \\
    \top & \bot & \bot & \bot & \bot \\
  \end{array}
 \]
 \caption{Examples of selection operators in a logic based on $\fourk$}\label{fig:selection-operators}
\end{figure}

Functional completeness also implies all unary selection operators can be expressed as formulas in the logics $L(\fourk,\set{\pnot,\bkor,\supset_{\rm b},\bot})$ and $L(\fourk,\set{\neg,\btand,\supset_{\rm b},\bot,\top})$.
However, we have the following result, from which it follows that neither of these logics is canonically complete.

 \begin{proposition}\label{pro:belnap-normal-form}
   $L(\fourk,\set{\neg, \btand, \btor, \bkand, \bkor, \supset_{\rm b}, \bot,0,1,\top})$ is not canonically complete.
 \end{proposition}

 \begin{proof}
 It is impossible to represent all unary selection operators in normal form. 
  The statement follows from the following observations:%
  \begin{inparaenum}[(i)]
   \item Belnap logic defines one unary operator $\pnot$;
   \item the only binary operators that may be used in normal form are $\bkor$ ($\minop$) and $\bkand$ ($\maxop$); and
   \item for any operator $\oplus \in \set{\pnot, \bkand,\bkor}$ we have $\bot \oplus \bot = \bot$.
  \end{inparaenum}
  Thus it is impossible to construct a unary operator of the form $\selop{\bot}{d}$ for any $d \ne \bot$.
 \end{proof} 

%
%
%

\begin{corollary}
PBel is not a canonically complete authorization language.
\end{corollary}

\begin{proof}
PBel uses the set of operators $\set{\bneg,\btand,\supset_{\rm b},\bot,\top}$, which is a subset of $\set{\neg, \btand, \btor, \bkand, \bkor, \supset_{\rm b}, \bot,0,1,\top}$.
Thus, by Proposition~\ref{pro:belnap-normal-form}, this is not a canonically complete set of operators, so $L(\fourk,\set{\bneg,\btand,\supset_{\rm b},\bot,\top})$ is not canonically complete either.
Hence, we may conclude that PBel is not a canonically complete authorization language.
\end{proof}


\section{A Canonically Complete\\ 4-valued Logic}
\label{sec:canon-suitable-set}

In the proof of Proposition~\ref{pro:belnap-normal-form}, we were unable to construct all unary selection operators using operators from the set $\set{\bneg,\bkand,\bkor}$, because there is no operator in which $\bot \oplus \bot \neq \bot$.
This suggests that we will require at least one additional unary operator $-$, say, such that $- \bot \ne \bot$.
Accordingly, we start with the unary operator, sometimes called ``conflation''~\cite{fitting:bilat91}, such that 
\[
 -\bot = \top,\ -\top = \bot,\ -0 = 0,\ \text{and } {-}1 = 1.
\]
Conflation is analogous to negation $\bnegt$, but inverts knowledge values rather than truth values.
In addition to $\bnegk$, we include the operator $\bkand$ in our set of operators, since this is the join operator for $\fourk$.

\begin{proposition}\label{pro:canonically-suitable-operators-for-belnap}
 $L(\fourk,\set{\bnegk,\bkand})$ is canonically suitable.
\end{proposition}

Informally, the proof follows from the fact that $\bnegk$ and $\bkand$ have exactly the same effect on $\fourk$ as $\neg$ and $\btand$ have on $(\four,\tord)$.
More formally, the following equivalence holds~\cite{arieli:four98}:
\[ 
d \bkor d' \equiv \bnegk (\bnegk d \bkand \bnegk d'). 
\]
The decision table establishing this equivalence is given in Figure~\ref{fig:bkor-encoding} (in Appendix~\ref{app:canonical-suitability-proof}).
Hence, we conclude that the set of operators $\set{\bnegk,\bkand}$ is canonically suitable, since $\minop$ corresponds to $\bkand$ and $\maxop$ corresponds to $\bkor$.

 
 \begin{proposition}
  $L(\fourk,\set{\bnegk,\bkand})$ is not functionally complete.
 \end{proposition}

 \begin{proof}
  The proof follows from the following observations:%
  \begin{inparaenum}[(i)]
   \item for the operators $\bnegk$ and $\bkand$, $\bnegk(0) = 0$ and $0 \bkand 0 = 0$; and
   \item any operator $\circ$ which is a combination of $\bnegk$ and $\bkand$, we have $0 \circ 0 = 0$.
  \end{inparaenum}
  Thus it is impossible to construct an operator in which $0 \circ 0 \ne 0$.
 \end{proof} 

To summarize: $L(\fourk,\set{\neg, \btand, \btor, \bkand, \bkor, \supset_{\rm b}, \bot,0,1,\top})$ is not canonically complete and $L(\fourk,\set{\bnegk,\bkand})$ is not functionally complete.
We now investigate what additional operators should be defined to construct a set of operators which is canonically complete (and hence functionally complete).

Given that we cannot use any operators besides $\maxop$ and $\minop$ in normal form, we focus on defining additional unary operators on $\fourk$.
An important observation at this point is that any permutation (that is, a bijection) $\pi : \four \rightarrow \four$ defines a unary operator on $\four$.
Accordingly, we now explore the connections between the group of permutations on $\four$ and unary operators on $\four$.
  
\subsection{The Symmetric Group and Unary\\ Operators}

  The \emph{symmetric group} $(S_X,\circ)$ on a finite set of $|X|$ symbols is the group whose elements are all permutations of the elements in $X$, and whose group operation $\circ$ is function composition.
  In other words, given two permutations $\pi_1$ and $\pi_2$, $\pi_1 \circ \pi_2$ is a permutation such that 
  \[
   (\pi_1 \circ \pi_2)(x) \stackrel{\rm def}{=} \pi_1(\pi_2(x)).
  \]
  We write $\pi^k$ to denote the permutation obtained by composing $\pi$ with itself $k$ times.
  
  A \emph{transposition} is a permutation which exchanges two elements and keeps all others fixed.
  Given two elements $a$ and $b$ in $X$, the permutation
   \[ \pi(x) =
     \begin{cases}
       b       & \quad \text{if } x = a, \\
       a & \quad \text{if } x = b, \\
       x & \quad \text{otherwise}, \\
     \end{cases}
   \]
  is a transposition, which we denote by $(a\, b)$.
  A \emph{cycle} of \emph{length} $k \geqslant 2$ is a permutation $\pi$ for which there exists an element $x$ in $X$ such that $x, \pi(x),\pi^2(x),\dots,\pi^k(x) = x$ are the only elements changed by $\pi$.
  Given $a$, $b$ and $c$ in $X$, for example, the permutation 
  \[ \pi(x) =
    \begin{cases}
      b       & \quad \text{if } x = a,\\
      c & \quad \text{if } x = b,\\
      a & \quad \text{if } x = c,\\
      x & \quad \text{otherwise},\\
    \end{cases}
  \]
  is a cycle of length $3$, which we denote by $(a\, b\, c)$.
  (Cycles of length two are transpositions.)
  The symmetric group $S_X$ is \emph{generated} by its cycles. 
  That is, every permutation may be represented as the composition of some combination of cycles.
  
  In fact, stronger results are known.
  We first introduce some notation.
  Let $X = \set{x_1,\dots,x_n}$ and let $S_n$ denote the symmetric group on the set of elements $\set{1,\dots,n}$.
  Then $(S_X,\circ)$ is trivially isomorphic to $(S_n,\circ)$ (via the mapping $x_i \mapsto i$).

  \begin{theorem}\label{thm:n-1-transpositions}
  For $n \geqslant 2$, $S_n$ is generated by the transpositions $(1\, 2),(1\, 3),\dots,(1\, n).$
  \end{theorem}
  
  \begin{theorem}\label{thm:ab-arb-transp-gen-set}
   For $1 \leqslant a < b \leqslant n$, the transposition $\transp{a}{b}$ and the cycle $(1\, 2\, \dots \, n)$ generate $S_n$ if and only if the greatest common divisor of $b-a$ and $n$ equals $1$.
  \end{theorem}
  
  In other words, it is possible to find a generating set comprising only transpositions, and it is possible to find a generating set containing only two elements.

  \subsection{New Unary Operators}
  
We now define three unary operators $\btran{0}, \btran{1}$ and $\btran{\top}$, which swap the value of $\bot$ and the truth value in the operator's subscript.
The truth tables for these operators are shown in Figure~\ref{fig:transpositions-btran}.
Note that $\btran{\top}$ is identical to the conflation operator $\bnegk$.
However, in the interests of continuity and consistency we will use the $\btran{\top}$ notation in the remainder of this section.

\begin{figure}[ht]
 \[
  \begin{array}{c|c|c|c}
   ~d~ & \btran{0} d~ & \btran{1} d & \btran{\top} d \\
  \hline
   \bot & 0 & 1 & \top \\
   0 & \bot & 0 & 0 \\
   1 & 1 & \bot & 1 \\
   \top & \top & \top & \bot\\
  \end{array}
 \]
 \caption{$\btran{0}, \btran{1}$ and $\btran{\top}$}\label{fig:transpositions-btran}
\end{figure}

Notice that $\btran{0}$, $\btran{1}$ and $\btran{\top}$ permute the elements of $\four$ and correspond to the transpositions $(\bot\, 0)$, $(\bot\, 1)$ and $(\bot\, \top)$, respectively.
Thus we have the following elementary result.

\begin{proposition}
\label{prop:any-permutation}
Any permutation on $\four$ can be expressed using only operators from the set $\set{\btran{0}, \btran{1},\btran{\top}}$.
\end{proposition}

\begin{proof}
The operators $\btran{0}, \btran{1}$ and $\btran{\top}$ are the transpositions $(\bot\, 0),(\bot\, 1)$ and $(\bot\, \top)$ respectively.
By Theorem~\ref{thm:n-1-transpositions}, these operators generate all the permutations in $S_\four$.
\end{proof}

\begin{lemma}
\label{lem:any-tuple-1234}
  It is possible to express any function $\phi : \four \rightarrow \four$ as a formula in $L(\fourk,\set{\btran{0}, \btran{1},\btran{\top},\bkand,\bkor})$.
 \end{lemma}

 \begin{proof}
  For convenience, we represent the function $\phi : \four \rightarrow \four$  as the tuple 
  \[
   (\phi(\bot),\phi(0),\phi(1),\phi(\top)) = (a,b,c,d).
  \]
  Then, given $x,y,z \in \four$, we define the function 
   \[ \phi_x^y(z) = %
     \begin{cases}
      x & \text{if $z = y$}, \\
      \bot & \text{otherwise}.
     \end{cases}
   \]
   Thus, for example, $\phi_a^{\bot} = (a,\bot,\bot,\bot)$.
   Then it is easy to see that for all $x \in \four$
   \[
    \phi(x) = \phi_a^{\bot}(x) \bkor \phi_b^{0}(x) \bkor \phi_c^{1}(x) \bkor \phi_d^{\top}(x).
   \]
   That is $\phi = \phi_a^{\bot} \bkor \phi_b^{0} \bkor \phi_c^{1} \bkor \phi_d^{\top}$.

   Thus, it remains to show that we can represent $\phi_a^{\bot}$, $\phi_b^{0}$, $\phi_c^{1}$ and $\phi_d^{\top}$ as formulas using the operators in $\set{\btran{0}, \btran{1},\btran{\top},\bkand,\bkor}$.
   First consider the permutations $\phi_{a,0}$, $\phi_{a,1}$ and $\phi_{a,\top}$, represented by the tuples $(a,\bot,b_1,c_1)$, $(a,b_2,\bot,c_2)$ and $(a,b_3,c_3,\bot)$, respectively.%
   \footnote{Note that the specific values of $b_i$ and $c_i$ are not important: it suffices that each of $\phi_{a,0}$, $\phi_{a,1}$ and $\phi_{a,\top}$ are permutations; once $b_i$ is chosen such that $b_i \not\in \set{a,\bot}$, then $c_i$ is fixed.}
   Since  $\phi_{a,0}$, $\phi_{a,1}$ and $\phi_{a,\top}$ are permutations, we know they can be written as some combination of the unary operators.
   Moreover, 
   \[
    \phi_a^{\bot} \equiv \phi_{a,0} \bkand \phi_{a,1} \bkand \phi_{a,\top}
   \]
   Clearly, we can construct $\phi_b^{0}$, $\phi_c^{1}$ and $\phi_d^{\top}$ in a similar fashion.
   The result now follows.
 \end{proof}

  The decision tables showing the construction of $\phi_a^{\bot}$ (column 5) and $\phi$ (column 10) are shown in Figure~\ref{fig:constructing-A-B-C-D}.

  \begin{figure*}[ht]\centering
     \[
      \begin{array}{c|c|c|c|c|c|c|c|c|c}
       x & \phi_{a,0} & \phi_{a,1} & \phi_{a,\top}& \phi_{a,0} \bkand \phi_{a,1} \bkand \phi_{a,\top} & \phi_a^{\bot} & \phi_b^{0} & \phi_c^{1} & \phi_d^{\top} & \phi_a^{\bot} \bkor \phi_b^{0} \bkor \phi_c^{1} \bkor \phi_d^{\bot}\\
      \hline
       \bot & a & a & a & a & a& \bot & \bot & \bot & a\\
       0 & \bot & b_2 & b_3 & \bot & \bot & b & \bot & \bot & b\\
       1 & b_1 & \bot & c_3 & \bot & \bot & \bot & c & \bot & c\\
       \top & c_1 & c_2 & \bot & \bot & \bot &  \bot & \bot & d & d
      \end{array}
  \]
  \caption{Expressing $\phi : \four \rightarrow \four$ using operators in $\set{\btran{0}, \btran{1},\btran{\top},\bkand,\bkor}$}\label{fig:constructing-A-B-C-D}
  \end{figure*}

  \begin{theorem}
  \label{thm:func-and-canon-complete}
  $L(\fourk,\set{\btran{0}, \btran{1},\btran{\top},\bkand,\bkor})$ is functionally and canonically complete.
  \end{theorem}

  \begin{proof}
  By Lemma~\ref{lem:any-tuple-1234}, it is possible to express any function $\phi : \four \rightarrow \four$ as a formula using operators from the set $\set{\btran{0}, \btran{1},\btran{\top},\bkand,\bkor}$.
  In particular, all unary selection operators can be expressed in this way.
  Hence by Theorem~\ref{thm:jobe-functional-completeness}, the set of operators $\set{\btran{0}, \btran{1},\btran{\top},\bkand,\bkor}$ is functionally complete.
  
  Moreover, all formulae constructed in the proof of Lemma~\ref{lem:any-tuple-1234} contain only the binary operators $\bkor (\maxop)$ and $\bkand (\minop)$, and unary operators defined as compositions of $\btran{0}, \btran{1}$ and $\btran{\top}$.
  Thus, by definition, the unary selection operators are in normal form.
  \end{proof}

  \begin{corollary}
   $L(\fourk,\set{\btran{0}, \btran{1},\btran{\top},\bkand})$ is functionally and canonically complete.
  \end{corollary}
  
  \begin{proof}
   The conflation operator $\bnegk$ and $\btran{\top}$ are identical.
   Hence \[ d \bkor d' \equiv \bnegk( \bnegk d \bkand \bnegk d')  \equiv \btran{\top}( \btran{\top} d \bkand \btran{\top} d'). \]
   Therefore, the set of operators is canonically suitable, and, by Theorem~\ref{thm:func-and-canon-complete}, it is functionally and canonically complete (since we can construct $\bkor$).
  \end{proof}

  \begin{corollary}\label{cor:minimal-set-of-operators-for-canonically-complete-4-valued-logic}
   Let $\diamond$ be the unary operator corresponding to the permutation given by the cycle $(\bot\, 0\, 1\, \top)$.
   Then $L(\fourk,\set{\btran{\top},\diamond,\bkand})$ is functionally and canonically complete.
  \end{corollary}

\begin{proof}
By Theorem~\ref{thm:ab-arb-transp-gen-set}, $\btran{\top}$ and $\bcyc$ generate all permutations in $S_\four$.
The remainder of the proof follows immediately from Lemma~\ref{lem:any-tuple-1234} and Theorem~\ref{thm:func-and-canon-complete}.
\end{proof}

%

It is important to note that we could choose any transposition $\transp{a}{b}$, such that $\mathsf{gcd}(b-a,n) = 1$.
We specifically selected the transposition $\transp{\bot}{\top}$, as this has the effect of reversing the minimum and maximum knowledge values.
Another choice for this transposition is one which swaps 0 and 1, specifically the transposition $\transp{0}{1}$.
This transposition is the truth negation operator $\bnegt$, which in the context of access control is a useful operator, since it swaps allow and deny decisions.

\subsection{Unary Operators for Totally Ordered\\ Logics}\label{sec:unary-ops-totally-ordered-logics}

Having shown the construction for a canonically complete 4-valued logic, in which the set of logical values forms a lattice, we briefly return to totally ordered logics.
We construct a totally ordered, canonically complete $m$-valued logic (thus extending the work of Jobe, who only showed how to construct a canonically complete $3$-valued logic).

Let $V$ be a totally ordered set of $m$ truth values, $\set{1,\dots,m}$, with $1 < \dots  < m$.
We define two unary operators $\mlogicneg$ and $\mlogiccycle$, which are the transposition $(1\, m)$ and the cycle $(1\,2\,\dots\,m)$, respectively.
In addition, we define one binary operator $\mand$, where $x \mand y = \max\set{x,y}$.
%

\begin{proposition}
\label{prop:any-permutation-total-ord}
Any permutation on $V$ can be expressed using only operators from the set $\set{\mlogicneg,\mlogiccycle}$.
\end{proposition}

\begin{proof}
The operator $\mlogicneg$ is the transposition $(1\, m)$ and the operator $\mlogiccycle$ is the cycle $(1\,2\,\dots\,m)$.
By Theorem~\ref{thm:ab-arb-transp-gen-set}, these operators generate all the permutations in $S_V$.
\end{proof}

\begin{proposition}
$L(V,\set{\mlogicneg,\mlogiccycle,\mand})$ is canonically suitable.
\end{proposition}

\begin{proof}
Clearly $x \minop y \equiv x \mand y$, it remains to show the operator $\maxop$ can be expressed in $L$. 
By Proposition~\ref{prop:any-permutation-total-ord} we can express any permutation of $V$ in terms of $\mlogicneg$ and $\mlogiccycle$.
In particular, we can express the permutation $f$, where $f(i) = m-i+1$, which swaps the values $1$ and $m$, $2$ and $m-1$, and so on.
We denote the unary operator which realizes this permutation by $\mrev$.
Then $x \maxop y \equiv x \mor y \equiv \mrev(\mrev x \mand \mrev y)$.
\end{proof}

\begin{theorem}
$L(V,\set{\mlogicneg,\mlogiccycle,\mand}$ is functionally and canonically complete.
\end{theorem}
%
%
%
We omit the proof, as it proceeds in an analogous manner to those for Lemma~\ref{lem:any-tuple-1234} and Theorem~\ref{thm:func-and-canon-complete}.
It is interesting to note that we have constructed a canonically complete $m$-valued logic which uses only two unary operators.
This is somewhat unexpected; intuition would suggest that $m-1$ unary operators are required for a canonically complete $m$-valued logic.

\section{A Canonically Complete\\ 4-valued ABAC Language}
\label{sec:canonically-complete-abac-lang-4val}

Having identified a canonically complete set of operators for Belnap logic, we now investigate how this set of operators can be used in an ABAC language, and consider the advantages in doing so.
Crampton and Williams~\cite{crampton:decision16} showed the operators in \ptacl can be replaced with an alternative set of operators, taken from Jobe's logic $E$, to obtain a canonically complete $3$-valued ABAC language.
In the remainder of this section, we describe a 4-valued lattice-ordered version of \ptacl, based on the lattice \mbox{$\fourk$}, which we denote by $\ptacll{\four}$.

\subsection{The Decision Set}

We first reiterate there is value in having an ABAC language for which policy evaluation can return a fourth value $\top$.
Such a value is used in both XACML and PBel, although its use in XACML is somewhat ad hoc and confusing since it can be used to indicate%
\begin{inparaenum}[(a)]
 \item an error in policy evaluation, or 
 \item a decision that arises for a particular operator during normal policy evaluation. 
\end{inparaenum}

We will use this fourth value to denote that (normal) policy evaluation has led to conflicting decisions (and we do not wish to use deny-overrides or similar operators to resolve the conflict at this point in the evaluation).
(We explain how we handle indeterminacy arising from errors in policy evaluation in Section~\ref{sec:indeterminacy}.)
Two specific operators, ``only-one-applicable'' (\textsf{ooa}) and ``unanimity'' (\textsf{un}) could make use of $\top$:
the \textsf{ooa} operator returns the value of the applicable sub-policy if there is only one such policy, and $\top$ otherwise; whereas the \textsf{un} operator returns $\top$ if the sub-policies return different decisions, and the common decision otherwise.
The decision tables for these operators are shown in Figure~\ref{fig:ooa-unanimity}.
 
\begin{figure}[!h]
\begin{subfigure}[b]{.45\columnwidth}\centering
 \[
  \begin{array}{c|cccc}
     \mathsf{ooa} & \bot & 0 & 1 & \top \\
   \hline
    \bot & \bot & 0 & 1 & \top \\
    0 &  0 & \top & \top & \top \\
    1 & 1 & \top & \top & \top \\
    \top & \top & \top & \top & \top \\
  \end{array}
 \]
\end{subfigure}
\hfill
\begin{subfigure}[b]{.45\columnwidth}\centering
 \[
  \begin{array}{c|cccc}
     \mathsf{un} & \bot & 0 & 1 & \top \\
   \hline
    \bot & \bot & \top & \top & \top \\
    0 & \top & 0 & \top & \top \\
    1 & \top & \top & 1 & \top \\
    \top & \top & \top & \top & \top \\
  \end{array}
 \]
\end{subfigure}
\caption{Operators using $\top$}\label{fig:ooa-unanimity}
\end{figure}  

In establishing canonical completeness for \ptacle, Crampton and Williams assumed a total order on the set of decisions ($0 < \bot < 1$).
This ordering does not really reflect the intuition behind the use of $0$, $1$ and $\bot$ in ABAC languages.
In the context of access control, $0$ and $1$ are incomparable conclusive decisions, and $\bot$ and $\top$ are decisions that reflect the inability to reach a conclusive decision either because a policy or its sub-policies are inapplicable ($\bot$) or because a policy's sub-policies return conclusive decisions that are incompatible in some sense ($\top$).
Moreover, we can subsequently resolve $\bot$ and $\top$ into one of two (incomparable) conclusive decisions using unary operators such as ``deny-by-default'' and ``allow-by-default''.
(The truth-based ordering on $\four$ does not correspond nearly so well to the above intuitions.)

\subsection{Operators and Policies}\label{sec:ptacl4-operators-policies}

We define the set of operators for $\ptacll{\four}$ to be $\set{\btran{\top},\bcyc,\bkand}$, which we established is canonically complete in Corollary~\ref{cor:minimal-set-of-operators-for-canonically-complete-4-valued-logic}.
Recall that $\btran{\top}$ is equivalent to conflation $-$; we will use the simpler notation $-$ in the remainder of this section.
An atomic policy has the form $(t,d)$, where $t$ is a target and $d \in \set{0,1}$.
(There is no reason for an atomic policy to return $\top$ -- which signifies a conflict has taken place -- in an atomic policy.)
Then we have the following policy semantics.
  \begin{align*}
   \semp{t,p}{q} &= %
    \begin{cases}
     \semp{p}{q} & \text{if $\semt{t}{q} = 1$}, \\
     \bot & \text{otherwise};
    \end{cases} \\
   \semp{d}{q} &= d; \\
   \semp{-p}{q} &= -\semp{p}{q}; \quad \semp{\bcyc p}{q} = \bcyc\semp{p}{q}; \\
   \semp{p \bkand p'}{q} &= \semp{p}{q} \bkand \semp{p'}{q}.
  \end{align*}

We now show how to represent the operator only-one-applicable ($\mathsf{ooa}$) in normal form.
(Recall that it is possible to represent this operator as a formula in PBel; however, it is non-trivial to derive such a formula.)
Using the truth table in Figure~\ref{fig:ooa-unanimity} and by definition of the selection operators and $\maxop$, we have \mbox{$x \ooa y$} is equivalent to
\begin{align*}
  &\selop{(\bot,\bot)}{\bot}(x,y) \maxop \selop{(\bot,0)}{0}(x,y) \maxop \selop{(\bot,1)}{1}(x,y) \maxop \selop{(\bot,\top)}{\top}(x,y) \maxop \\
  & \selop{(0,\bot)}{0}(x,y) \maxop \selop{(0,0)}{\top}(x,y) \maxop \selop{(0,1)}{\top}(x,y) \maxop \selop{(0,\top)}{\top}(x,y) \maxop \\
  & \selop{(1,\bot)}{1}(x,y) \maxop \selop{(1,0)}{\top}(x,y) \maxop \selop{(1,1)}{\top}(x,y) \maxop \selop{(1,\top)}{\top}(x,y) \maxop \\
  & \selop{(\top,\bot)}{\top}(x,y) \maxop \selop{(\top,0)}{\top}(x,y) \maxop \selop{(\top,1)}{\top}(x,y) \maxop \selop{(\top,\top)}{\top}(x,y).
\end{align*}
Moreover, $\selop{(x,y)}{z} = \selop{x}{z} \minop \selop{y}{z}$ and $\selop{x}{y}$ is a function $\phi : \four \rightarrow \four$, which can be represented as a composition of unary operators.
Hence, we can derive a formula  in normal form for $\ooa$.

Functional completeness implies we can write any binary operator (such as XACML's deny-overrides policy-combining algorithm) as a formula in $L(\fourk,\set{-,\bcyc,\bkand}$, and hence we can use any operator we wish in $\ptacll{\four}$ policies.
However, canonical completeness and the decision set $(\four,\kord)$ allows for a completely different approach to constructing ABAC policies.
Suppose a policy administrator has identified three sub-policies $p_1$, $p_2$ and $p_3$ and wishes to define an overall policy $p$ in terms of the decisions obtained by evaluating these sub-policies.
Then the policy administrator can tabulate the desired decision for all relevant combinations of decisions for the sub-policies, as shown in the table below.
The default decision is to return $\bot$, indicating that $p$ is ``silent'' for other combinations.
\[
 \begin{array}{ccc|c}
  p_1 & p_2 & p_3 & p \\
 \hline
  \bot & 0 & 0 & 0 \\
  0 & 0 & 0 & 0 \\
  1 & 0 & 0 & \top \\
  1 & 1 & 0 & 1 \\
  1 & 1 & 1 & 1 \\
 \end{array}
\]
Then, treating $p$ as a function of its sub-policies, we have
\[
 p \equiv\ \selop{(\bot,0,0)}{0} \maxop \selop{(0,0,0)}{0} \maxop \selop{(1,0,0)}{\top} \maxop \selop{(1,1,0)}{1} \maxop \selop{(1,1,1)}{1}. 
\]
The construction of $p$ as a ``disjunction'' ($\maxop$) of selection operators ensures that the correct value is returned for each combination of values (in much the same as disjunctive normal form may be used to represent the rows in a truth table).
Note that if $p_1$, $p_2$ and $p_3$ evaluate to a tuple of values other than one of the rows in the table, each of the selection operators will return $\bot$ and thus $p$ will evaluate to $\bot$.
Each operator of the form $\selop{(a,b,c)}{d}$ can be represented as the ``conjunction'' ($\minop$) of unary selection operators (specifically $\selop{a}{d} \minop \selop{b}{d} \minop \selop{c}{d}$).

Of course, one would not usually construct the normal form by hand, as we have done above.
Indeed, we have developed an algorithm which takes an arbitrary policy expressed as a decision table as input, and outputs the equivalent normal form expressed in terms of the operators $\set{-,\bcyc,\bkand}$.
In order to develop this algorithm, we also derived expressions for the unary selection operators in terms of the operators $\set{-,\bcyc,\bkand}$.
(In Lemma~\ref{lem:any-tuple-1234} we only showed that such expressions exist.)
Our implementation of the algorithm, comprising just less than $150$ lines of Python code, shows the ease with which construction of policies can be both automated and simplified, utilizing the numerous advantages that have been discussed throughout this paper.%
\footnote{Code and test results available at \url{goo.gl/0TM0RD}.}

%
\vfill
\subsection{Indeterminacy}\label{sec:indeterminacy}

XACML uses the indeterminate value in two distinct ways:
\begin{enumerate}
 \item as a decision returned (during normal evaluation) by the ``only-one-applicable'' policy-combining algorithm; and
 \item as a decision returned when some (unexpected) error has occurred in policy evaluation has occurred.
\end{enumerate}
In the second case, the indeterminate value is used to represent alternative outcomes of policy evaluation (had the error not occurred).
We believe that the two situations described are quite distinct and require different policy semantics.
However, the semantics of indeterminacy in XACML are confused because%
\begin{inparaenum}[(i)]
 \item the indeterminate value is used in two different ways, as described above, and
 \item there is no clear and uniform way of establishing the values returned by the combining algorithms when an indeterminate value is encountered.
\end{inparaenum}

We have seen how $\top$ may be used to represent decisions for operators such as $\ooa$ and $\mathsf{un}$.
%
We handle errors in target evaluation (and thus indeterminacy) using sets of possible decisions~\cite{huth:auth10,crampton:ptacl12,li:access09}.
(This approach was adopted in a rather ad hoc fashion in XACML 3.0, using an extended version of the indeterminate decision.)
Informally, when target evaluation fails, denoted by $\semt{t}{q} =\ ?$, \ptacl assumes that either $\semt{t}{q} = 1$ or $\semt{t}{q} = 0$ could have been returned, and returns the union of the (sets of) decisions that would have been returned in both cases.
The formal semantics for policy evaluation in $\ptacll{\four}$ in the presence of indeterminacy are defined in Figure~\ref{fig:ptacl-indeterminacy}.

\begin{figure}[h]
  \begin{align*}
   \semp{t,p}{q} &= %
    \begin{cases}
     \semp{p}{q} & \text{if $\semt{t}{q} = 1$}, \\
     \set{\bot} & \text{if $\semt{t}{q} = 0$}, \\
     \set{\bot} \cup \semp{p}{q} & \text{if $\semt{t}{q} =\ ?$},
    \end{cases} \\
   \semp{d}{q} &= \set{d}; \\
   \semp{-p}{q} &= \set{ -d: d \in \semp{p}{q}}; \\
   \semp{\bcyc p}{q} &= \set{ \bcyc d: d \in \semp{p}{q}}; \\
   \semp{p_1 \bkand p_2}{q} &= \set{d_1 \bkand d_2 : d_i \in \semp{p_i}{q}}.
  \end{align*}
 \caption{Semantics for $\ptacll{\four}$ with indeterminacy}\label{fig:ptacl-indeterminacy}
\end{figure}

The semantics for the operators $\set{-,\bcyc,\bkand}$ operate on sets, rather than single decisions, in the natural way.
A straightforward induction on the number of operators in a policy establishes that the decision set returned by these extended semantics will be a singleton if no target evaluation errors occur; moreover, that decision will be the same as that returned by the standard semantics.

\subsection{Leveraging the XACML Architecture}

XACML is a well-known, standardized language, and many of the components and features of XACML are well-defined.
However, it has been shown that the rule- and policy-combining algorithms defined in the XACML standard suffer from some shortcomings~\cite{li:access09}, notably inconsistencies between the rule- and policy-combining algorithms.
\ptacl, on which $\ptacll{\four}$ is based, is a tree-structured ABAC language that is explicitly designed to use the same general policy structure and evaluation methods as XACML.
However, \ptacl differs substantially from XACML in terms of policy combination operators and semantics.

Thus, we suggest that $\ptacll{\four}$ operators could replace the rule- and policy-combining algorithms of XACML, while those parts of the language and architecture that seem to function well may be retained.
Specifically, we use the XACML architecture to:%
 \begin{inparaenum}[(i)]
   \item specify requests;
   \item specify targets;
   \item decide whether a policy target is applicable to a given request; and
   \item use the policy decision point to evaluate policies.
  \end{inparaenum}
In addition, we would retain the enforcement architecture of XACML, in terms of the policy decision, policy enforcement and policy administration points, and the relationships between them.

We believe it would be relatively easy to modify the XACML PDP to
\begin{itemize}
 \item handle four decisions, extending the current set of values (``allow'', ``deny'' and ``not-applicable'') to include ``conflict'';
 \item implement the policy operators $\set{-,\bcyc,\bkand}$ as custom combining algorithms; and
 \item work with decision sets, in order to handle indeterminacy in a uniform manner.
\end{itemize}
For illustrative purposes, Appendix~\ref{app:A} specifies the modified decision set and pseudocode for the operator $\bkand$ in the format used by the XACML standard.

The main difference to end-users would be in the simplicity of policy authoring.
Using standard XACML, policy authors must decide which rule- and policy-combining algorithms should be used to develop a policy or policy set that is equivalent to the desired policy.
This is error-prone and it may not even be possible to express the desired policy using only the XACML combining algorithms.
Using XACML with the policy-combining mechanisms of $\ptacll{\four}$, we can present an entirely different interface for policy authoring to the end-user.
The policy author would first specify the atomic policies (XACML rules), then combine atomic policies using decision tables to obtain more complex policies (as illustrated in Section~\ref{sec:ptacl4-operators-policies}).
Those policies can be further combined by specifying additional decision tables.
At each stage a back-end policy compiler can be used to convert those policies into policy sets (using $\ptacll{\four}$ operators) that can be evaluated by the XACML engine.

 \section{Concluding Remarks}

Attribute-based access control is of increasing importance, due to the increasing use of open, distributed, interconnected and dynamic systems.
The introduction of \emph{canonically complete} ABAC languages~\cite{crampton:decision16} provides the ability to express any desired policy in a normal form, which allows for the possibility of specifying policies in the form of a decision table and then automatically compiling them into the language.

In this paper, we make important contributions to the understanding of canonical completeness in multi-valued logics and thus in ABAC languages.
First, we extend Jobe's work on canonical completeness to multi-valued logics to the case where the set of truth values forms a lattice.
We show that the Belnap set of operators~\cite{belnap:logic77} (and thus any subset thereof) is not canonically complete, hence any ABAC language based on these operators cannot be canonically complete.
In particular, PBel~\cite{bruns:belnap11}, probably the most well known $4$-valued ABAC language, is not canonically complete.
We introduce a new four-valued logic $L(\fourk,\set{-,\bcyc,\bkand})$ which is canonically complete, without having to explicitly construct the unary selection operators in normal form (unlike Jobe~\cite{jobe:functional62} and Crampton and Williams~\cite{crampton:decision16}).
By identifying the connection between the generators of the symmetric group and the unary operators of logics, we have developed a simple and generic method for identifying a set of unary operators that will guarantee the functional and canonical completeness of an $m$-valued lattice-based logic.
We also showed that there is a set of operators containing only three connectives which is functionally complete for Belnap logic, in contrast to the set of size four identified by Arieli and Avron~\cite{arieli:four98}.

Second, we show in $\ptacll{\four}$ how the canonically complete set of operators $\set{\bnegk,\bcyc,\bkand}$ can be used in an ABAC language, and present the advantages of doing so.
In particular, we are no longer forced to use a totally ordered set of three decisions to obtain canonical completeness (as in the case in \ptacle).
Moreover, the overall design of \ptacl and hence $\ptacll{\four}$ is compatible with the overall structure of XACML policies.
We discuss how the XACML decision set and rule-combining algorithms can be modified to support $\ptacll{\four}$.
Doing so enables us to retain the rich framework provided by XACML for ABAC (in terms of its languages for representing targets and requests) and its enforcement architecture (in terms of the policy enforcement, policy decision and policy administration points).
Thus, we are able to propose an enhanced XACML framework within which any desired policy may be expressed.
Moreover, the canonical completeness of $\ptacll{\four}$, means that the desired policy may be represented in simple terms by a policy author (in the form of a decision table) and automatically compiled into a PDP-readable equivalent policy.

Our work paves the way for a considerable amount of future work.
In particular, we intend to develop a modified XACML PDP that implements the $\ptacll{\four}$ operators.
We also hope to develop a policy authoring interface in which users can simply state what decision a policy should return for particular combinations of decisions from sub-policies.
This would enable us to evaluate the usability of such an interface and compare the accuracy with which policy authors can generate policies using standard XACML combining algorithms compared with the methods that $\ptacll{\four}$ can support.

On the more technical side, we would like to revisit the notion of \emph{monotonicity}~\cite{crampton:ptacl12} in targets and how this affects policy evaluation in ABAC languages.
The definition of monotonicity is dependent on the ordering chosen for the decision set and existing work on monotonicity assumes the use of a totally ordered $3$-valued set (comprising $0$, $\bot$ and $1$).
So it will be interesting to consider how the use of a $4$-valued lattice-ordered decision set affects monotonicity.
We also intend to investigate methods of \emph{policy compression}, analogous to the minimization of Boolean functions~\cite{mccl:mini56}, where we take the canonical form of a policy (generated from a decision table) and rewrite it in such a way as to minimize the number of terms in the policy.

\appendix

\section{Jobe's Canonically Complete 3-valued logic}
\label{app:jobe-3-val}

Consider the 3-valued logic $J$~\cite{jobe:functional62}, whose operators $\jand,\btran{1}$ and $\btran{2}$ are defined in Figure~\ref{table:jobelogic}.
\begin{figure}[h]\centering
   \[
    \begin{array}{c|c|c}
     x & \btran{1} x & \btran{2} x \\
    \hline
     0 & 1 & 2 \\
     1 & 0 & 1 \\
     2 & 2 & 0
    \end{array}
    \qquad
     \begin{array}{c|ccc}
      \jand & 0 & 1 & 2 \\
     \hline
      0 & 0 & 0 & 0 \\
      1 & 0 & 1 & 1 \\
      2 & 0 & 1 & 2 
     \end{array}
   \]
\caption{The operators in Jobe's logic}\label{table:jobelogic}
\end{figure}

It is easy to establish that
\[ 
x \minop y \equiv x \jand y \quad\text{and}\quad x \maxop y \equiv \btran{2}(\btran{2}(x) \jand \btran{2}(y)). 
\]  
Thus $J$ is canonically suitable~\cite[Theorem 6]{jobe:functional62}.
The normal-form formulas for the unary selection operators are shown in Figure~\ref{fig:normal-form-selection-ops-in-E}.
(Note that $\selop{i}{0}$ is the same for all $i$.)
Thus $J$ is functionally and canonically complete~\cite[Theorem 7]{jobe:functional62}.
Hence, it is possible to construct a canonically complete $3$-valued logic using the operators $\set{\jand,\sim_0,\sim_1}$.

\begin{figure}[h]\centering
 \[
  \begin{array}{|l|l|}
  \hline
   \selop{i}{0}(x) & x \jand \btran{1}(x) \jand \btran{2}(x) \\
  \hline
   \selop{0}{1}(x) & \btran{1}(x) \jand \btran{2}\btran{1}(x) \\
   \selop{1}{1}(x) & x \jand \btran{2}(x) \\
   \selop{2}{1}(x) & \btran{1}\btran{2}(x) \jand \btran{2}\btran{1}\btran{2}(x) \\
  \hline
   \selop{0}{2}(x) & \btran{2}(x) \jand \btran{1}\btran{2}(x) \\
   \selop{1}{2}(x) & \btran{2}\btran{1}(x) \jand \btran{2}\btran{1}\btran{2}(x) \\
   \selop{2}{2}(x) & x \jand \btran{1}(x) \\
  \hline
  \end{array}
 \]
 \caption{Normal forms for the unary selection operators}\label{fig:normal-form-selection-ops-in-E}
\end{figure}

\section{Operators in Belnap Logic}\label{app:belnap-ops}

\begin{figure}[th!]
\begin{subfigure}[b]{.4\columnwidth}\centering
 \[
  \begin{array}{c|cccc}
     \btand & 0 & \bot & \top & 1 \\
   \hline
    0 & 0 & 0 & 0 & 0 \\
    \bot & 0 & \bot & 0 & \bot \\
    \top & 0 & 0 & \top & \top \\
    1 & 0 & \bot & \top & 1 \\
  \end{array}
 \]
 \caption{$\btand$}\label{fig:belnap-truth-and}
\end{subfigure}
\hfill
\begin{subfigure}[b]{.4\columnwidth}\centering
 \[
  \begin{array}{c|cccc}
     \btor & 0 & \bot & \top & 1 \\
   \hline
    0 & 0 & \bot & \top & 1 \\
    \bot & \bot & \bot & 1 & 1 \\
    \top & \top & 1 & \top & 1 \\
    1 & 1 & 1 & 1 & 1 \\
  \end{array}
 \]
 \caption{$\btor$}\label{fig:belnap-truth-or}
\end{subfigure}
\begin{subfigure}[b]{.4\columnwidth}\centering
 \[
  \begin{array}{c|cccc}
     \bkand & \bot & 0 & 1 & \top \\
   \hline
    \bot & \bot & \bot & \bot & \bot \\
    0 & \bot & 0 & \bot & 0 \\
    1 & \bot & \bot & 1 & 1 \\
    \top & \bot & 0 & 1 & \top \\
  \end{array}
 \]
 \caption{$\bkand$}\label{fig:belnap-know-and}
\end{subfigure}
\hfill
\begin{subfigure}[b]{.4\columnwidth}\centering
 \[
  \begin{array}{c|cccc}
  \bkor & \bot & 0 & 1 & \top \\
   \hline
    \bot & \bot & 0 & 1 & \top \\
    0 & 0 & 0 & \top & \top \\
    1 & 1 & \top & 1 & \top \\
    \top & \top & \top & \top & \top \\
  \end{array}
 \]
 \caption{$\bkor$}\label{fig:belnap-know-or}
\end{subfigure}

\centering
\begin{subfigure}[b]{.4\columnwidth}\centering
 \[
  \begin{array}{c|cccc}
     \supset_{\rm b} & 0 & \bot & \top & 1 \\
   \hline
    0 & 1 & 1 & 1 & 1 \\
    \bot & 1 & 1 & 1 & 1 \\
    \top & 0 & \bot & \top & 1 \\
    1 & 0 & \bot & \top & 1 \\
  \end{array}
 \]
 \caption{$\supset_{\rm b}$}\label{fig:belnap-supset}
\end{subfigure}
\hfill
\begin{subfigure}[b]{.4\columnwidth}\centering
 \[
  \begin{array}{c|r}
     d & ~\dnot d  \\
   \hline
    0 & 1  \\
    \bot & \bot \\
    \top & \top \\
     1 & 0 
  \end{array}
 \]
 \caption{$\neg$}\label{fig:belnap-negation}
\end{subfigure}
\caption{Operators in Belnap logic}\label{fig:belnap-operators}
\end{figure}

\section{Proof of Proposition~2}\label{app:canonical-suitability-proof}

The decision table in Figure~\ref{fig:bkor-encoding} establishes the equivalence of $x \bkor y$ and $-(x \bkand -y)$, which proves that \mbox{$L((\four,\kord),\set{-,\bkand})$} is a canonically suitable logic (Proposition~\ref{pro:canonically-suitable-operators-for-belnap}).

 \begin{figure}[ht]\centering
  \[
   \begin{array}{ r | r | r | r | r | r | r}                       
      d & ~d' & \bnegk d & \bnegk d'  & \bnegk d \bkand \bnegk d' & \bnegk (\bnegk d \bkand \bnegk d') & d \bkor d' \\
     \hline
		\bot & \bot & \top & \top & \top & \bot & \bot \\
		\bot & 0 & \top & 0 & 0 & 0 & 0 \\
		\bot & 1 & \top & 1 & 1 & 1 & 1 \\
		\bot & \top & \top & \bot & \bot & \top & \top \\
		0 & \bot & 0 & \top & 0 & 0 & 0 \\
		0 & 0 & 0 & 0 & 0 & 0 & 0 \\
		0 & 1 & 0 & 1 & \bot & \top & \top \\
		0 & \top & 0 & \bot & \bot & \top & \top \\
		1 & \bot & 1 & \top & 1 & 1 & 1 \\
		1 & 0 & 1 & 0 & \bot & \top & \top \\
		1 & 1 & 1 & 1 & 1 & 1 & 1 \\
		1 & \top & 1 & \bot & \bot & \top & \top \\
		\top & \bot & \bot & \top & \bot & \top & \top \\
		\top & 0 & \bot & 0 & \bot & \top & \top \\
		\top & 1 & \bot & 1 & \bot & \top & \top \\
		\top & \top & \bot & \bot & \bot & \top & \top \\
		
   \end{array} 
  \]
  
\caption{Encoding $\bkor$ using $\bnegk$ and $\bkand$}\label{fig:bkor-encoding}
\end{figure}

\vfill

\section{Encoding PT{a}CL Decisions and\\ Operators}
\label{app:A}

In Figures~\ref{fig:ptacl-decisions-in-xacml} and~\ref{fig:ptacl-ops-in-xacml} we illustrate how $\ptacll{\four}$ extensions could be incorporated in XACML by encoding the $\ptacll{\four}$ decisions and $\bkand$ operator using the syntax of the XACML standard.

\begin{figure}[h]
{\scriptsize
\begin{lstlisting}
<xs:element name=``Decision'' 
	    type=``xacml:DecisionType''/>
<xs:simpleType name=``DecisionType''>
 <xs:restriction base=``xs:string''>
   <xs:enumeration value=``Permit''/>
   <xs:enumeration value=``Deny''/>
   <xs:enumeration value=``Conflict''/>
   <xs:enumeration value=``NotApplicable''/>
 </xs:restriction>
</xs:simpleType>
\end{lstlisting}}
\caption{The $\ptacll{\four}$ decision set in XACML syntax}\label{fig:ptacl-decisions-in-xacml}
\end{figure}

\begin{figure}[h]
{\scriptsize
\begin{lstlisting}
Decision ptaclCombiningAlgorithm(Node[] children)
{
 Boolean atLeastOneDeny  = false;
 Boolean atLeastOnePermit  = false;
 for( i=0 ; i < lengthOf(children) ; i++ )
 {
   Decision decision = children[i].evaluate();
   if (decision = = NotApplicable)
   {  return NotApplicable;  }
   if (decision = = Permit)
   {
     atLeastOnePermit = true;
     continue;
   }
   if (decision = = Deny)
   {
     atLeastOneDeny = true;
     continue;
   }
   if (decision = = Conflict)
   {  continue;  }
 }
 if (atLeastOneDeny & & atLeastOnePermit)
 {  return NotApplicable;  }
 if (atLeastOneDeny)
 {  return Deny;  }
 if (atLeastOnePermit)
 {  return Permit;  }
 return Conflict;
}
\end{lstlisting}
}
\caption{The $\ptacll{\four}$ operator $\bkand$ encoded as an XACML combining algorithm}\label{fig:ptacl-ops-in-xacml}
\end{figure}  

\end{document}